\documentclass{amsart}
\usepackage[mathscr]{eucal}
\usepackage{amssymb, amsmath,array, amscd}
\usepackage[OT2,T1]{fontenc}
\usepackage[fontsize=11pt]{scrextend}
\usepackage{graphicx}
\usepackage{dcolumn}
\usepackage{bm}
\usepackage{graphics}
\usepackage{slashed}
\usepackage{amssymb}
\usepackage{amsmath}
\usepackage{amsthm}
\usepackage{url}
\usepackage[usenames,dvipsnames,svgnames,table]{xcolor}
\usepackage[T1]{fontenc}
\DeclareFontFamily{T1}{calligra}{}
\DeclareFontShape{T1}{calligra}{m}{n}{<->s*[1.44]callig15}{}
\DeclareMathAlphabet\mathcalligra   {T1}{calligra} {m} {n}
\DeclareMathAlphabet\mathzapf       {T1}{pzc} {mb} {it}
\DeclareMathAlphabet\mathchorus     {T1}{qzc} {m} {n}
\DeclareMathAlphabet\mathrsfso      {U}{rsfso}{m}{n}
\newcommand{\bea}{\begin{eqnarray}}
\newcommand{\ena}{\end{eqnarray}}
\newcommand{\bean}{\begin{eqnarray*}}
\newcommand{\enan}{\end{eqnarray*}}

\newcommand{\transv}{\mathrel{\text{\tpitchfork}}}
\makeatletter
\newcommand{\tpitchfork}{%
  \vbox{
    \baselineskip\z@skip
    \lineskip-.52ex
    \lineskiplimit\maxdimen
    \m@th
    \ialign{##\crcr\hidewidth\smash{$-$}\hidewidth\crcr$\pitchfork$\crcr}
  }%
}
\makeatother

\newtheorem{thm}{Theorem}

\newtheorem{prop}{Proposition}
 
\newtheorem{defn}{Definition}

\begin{document}

\title[spacetime foliations]{On spacetime foliations and electromagnetic knots}


\author[W. Costa e Silva]{W. Costa e Silva}
\address{IMPA, Est. D. Castorina, 110, 22460-320, Rio de Janeiro, RJ, Brazil}

\email{wancossil@gmail.com}

\author[E. Goulart]{E. Goulart}
\address{Departamento de F\'isica e Matem\'atica 
CAP-Universidade Federal de S\~ao Jo\~ao del-Rei, Rod. MG 443, KM 7, 36420-000, Ouro Branco, MG, Brazil}

\email{egoulart2@gmail.com}

\author[J. E. Ottoni]{J. E. Ottoni}
\address{Departamento de F\'isica e Matem\'atica 
CAP-Universidade Federal de S\~ao Jo\~ao del-Rei, Rod. MG 443, KM 7, 36420-000, Ouro Branco, MG, Brazil}

\email{jeottoni@gmail.com }



\begin{abstract}
The present paper has a number of distinct purposes. First is to give a description of a class of electromagnetic knots from the perspective of foliation theory. Knotted solutions are then interpreted in terms of two codimension-2 foliations whose knotted leaves intersect orthogonally everywhere in spacetime. Secondly, we show how the foliations give rise to field lines and how the topological invariants emerge. The machinery used here emphasizes intrinsic properties of the leaves instead of observer dependent quantities - such as a time function, a local rest frame or a  Cauchy hypersurface. Finally, we discuss the celebrated Hopf-Ra\~nada solution in details and stress how the foliation approach may help in future developments of the theory of electromagnetic knots. We conclude with several possible applications, extensions and generalizations.\end{abstract}

\maketitle

\section{Introduction}\label{Section I}

The infusion of topological techniques into field theory has led physicists to investigate knotted configurations in a plethora of physical systems. Possible applications of such knotty textures range from particle physics to astrophysical plasmas and include a variety of soft matter systems such as quantum fluids, liquid crystals, superconductors and ferromagnetism \cite{kedia0}. Interestingly, recent technological advances have permitted laboratory realizations of knots in the vortex lines of a fluid using 3D-printed wings, in the director field in nematic liquid crystals and in the lines of darkness in optical beams \cite{dennis,kleckner,Tkalec}, whereas nontrivial magnetic surfaces in stellarators have now been verified with high accuracy \cite{sunn}.

Mathematically, a knot is the embedding of a circle in a 3-manifold and a collection of such embeddings fulfilling nicely the latter is called a fibration. Examples include the Hopf and Seifert fibrations of the 3-sphere, well known to low-dimensional topologists \cite{ghys,milnor,rolfsen,scott}. These entities are also common place in the theory of topological solitons and play an important role in nonlinear models such as the Faddeev-Niemi model of glueballs \cite{manton}. Less well known, however, is the result that Maxwell equations \textit{in vacuo} admit solutions such that all lines of force twist in space generating knotted pencils of curves which may be linked in a variety of ways. These solutions are called \textit{electromagnetic knots} and the appearance of them in an eminently linear theory is quite remarkable from the theoretical point of view (see \cite{bouwmeester} for a recent review).

The genesis of electromagnetic knots dates back (at least) to Trautman, who investigated aspects of gauge field equations associated with Hopf fiberings \cite{trau}. However, only after Ra\~nada's work series \cite{ratru,ratru1,ran} the area gained the necessary impetus to become an independent field of research. Ra\~nada's technology is based on maps from Minkowski spacetime to the complex plane $(\phi,\theta):\mathbb{R}^{1+3}\rightarrow\mathbb{C}$ which are required to be single-valued at spatial infinity. After identifying the complex plane with the unit sphere via stereographic projection, he effectively obtains $(\phi,\theta):\mathbb{R}\times\mathbb{S}^{3}\rightarrow \mathbb{S}^{2}$, which are nontrivial, since the third homotopy group of the 2-sphere is $\pi_{3}(\mathbb{S}^{2})\simeq\mathbb{Z}$. The Hopf-Ra\~nada solution has been re-discovered in different contexts using, for instance, twistor and spinor techniques \cite{dal}.

Recently, several new electromagnetic knots have been reported. Attempts at generalizing the Hopfion to fields which mimic the Seifert fibration were derived in \cite{at1,at2,at3}, but the authors showed that the topological structure was not preserved, and unraveled with time. More persistent are the solutions obtained using a blend of Bateman's formalism, spinors and complex polynomials. These are null fields by construction and evolve as though field lines are embedded unbreakable filaments in a fluid flowing at the speed of light, in the direction of the Poynting vector. Examples of the latter include all (p,q) torus knots solutions and several types of links. Another approach uses conformal transformations with complex parameters, and a formulation of the latter in terms of quaternions has led to a map between the electromagnetic knots into flat connections of $SU(2)$ gauge theory \cite{hoyos}.

The present paper is motivated by the following verification: most previous attempts in investigating electromagnetic knots were grounded on the concept of field lines. Although this approach is mathematically simpler and aesthetically pleasant it is somehow unsatisfactory from the relativist perspective. The reasons for this are the following: i) field lines require the introduction of artificial structures - such as a time function or a global reference frame - that have no intrinsic relation to either the spacetime geometry or the particular electromagnetic field being discussed. ii) Maxwell's equations written in terms of electric and magnetic fields become highly involved with the kinematical parameters of the observer field and therefore, it seems that the geometry/topology of the lines may look quite different for different observers. iii) for some fields, such as a magnetically dominated solution, it is possible to find an observer such that no electric field is present. In such cases it is meaningless to talk about the topology of field lines since some observers `detect' the field while the others not.

In order to bypass these difficulties we investigate the electromagnetic knots from the perspective of foliation theory. To do so we introduce the concept of a dual foliation and stress some of its elementary properties. From our perspective, knotted configurations become two codimension-2 spacetime foliations which extend the notion of a field line to a field sheet. Such field sheets have no relation to a particular field of observers and have an intrinsic geometrical/topological meaning. We believe this approach is particularly illuminating and paves the way for future related problems such as extension to curved backgrounds, manifolds with nontrivial topologies and may serve as a unifying method in the construction of new knotted solutions. It also provides nontrivial examples of foliations which were not fully appreciated, up to now, by mathematicians. However, our construction solves the problem partially, since it applies to decomposable (simple) electromagnetic fields i.e., those which are the wedge product of 1-forms. Extension to generic fields such as those studied in \cite{at1,at2,at3} is important and will be the focus of future research.

The structure of this paper is as follows: Section \ref{Section II} starts placing definitions in the setting of foliation theory \cite{lawson,camlins,lee} and it is quite general. In section \ref{Section III}, we introduce the electromagnetic fields in vacuum and define the helicities in a covariant way. Section \ref{Section IV} shows how \textit{simple} solutions define dual foliations and how to classify them. Afterwards, fields defined by submersions are introduced and it is also shown how field lines relate to the concept of a dual foliation. Section \ref{Section V} applies the machinery of dual foliations to the case of the Hopf-Ra\~nada solution and reobtains several known results in a fully covariant fashion. We conclude with possible applications and extensions of the theory.  

\section{Distributions and dual foliations}\label{Section II}

Let $M$ be a smooth $m$-dimensional manifold and $D\subseteq TM$ a \textit{rank-}$k$ \textit{distribution} spanned by smooth vector fields $V_{1},...,V_{k}:M \rightarrow TM$. We say that $D$ is \textit{involutive} if the Lie bracket $[V_{i},V_{j}]$ is a section of $D$ for each $i,j=1,...,k$. Alternatively, we can define a rank-$k$ distribution in terms of smooth 1-forms $\omega^{1},...,\omega^{m-k}$ such that, for each $p\in M$
\begin{equation}\label{Ker}
D_{p}=\mbox{ker}\ \omega^{1}(p)\cap ... \cap \mbox{ker}\ \omega^{m-k}(p).
\end{equation}
Any $m-k$ linearly independent 1-forms $\omega^{1},...,\omega^{m-k}$ defined on an open subset $U\subseteq M$ and satisfying (\ref{Ker}) for each $p\in U$ are called \textit{local defining forms for} $D$. It can be shown that $D$ is involutive on $U$ if and only if 
\begin{equation}
d\omega^{i}\wedge\omega^{1}\wedge...\wedge\omega^{m-k}=0
\end{equation}
holds for each $i=1,...,m-k$ and \textit{Frobenius theorem} guarantees that, in this case, the distribution is completely integrable. This means that, when we put together all of the maximal integral manifolds of an involutive rank-$k$ distribution $D$, we obtain a $k$-dimensional foliation of $M$. 

\begin{defn}
By a $k$-dimensional, class $C^{r}$ foliation of an $m$-dimensional manifold $M$ we mean a decomposition of $M$ into a disjoint union of connected subsets $\{\mathrsfso{L}_{\alpha}\}_{\alpha\in A}$, called the \textit{leaves} of the foliation, with the following property: Every point $p\in M$ has a neighborhood $U$ and a system of local, class $C^{r}$ coordinates $x^{a}=(x^{1},...,x^{m}): U\rightarrow \mathbb{R}^{m}$ such that for each leaf $\mathrsfso{L}_{\alpha}$, the components of $U\cap\mathrsfso{L}_{\alpha}$ are described by the equations $x^{k+1}=constant^{k+1},...,x^{m}=constant^{m}$. 
\end{defn}

We call the partition $\mathrsfso{F}=\{\mathrsfso{L}_{\alpha}\}_{\alpha\in A}$, a codimension $m-k$ foliation of $M$ and the coordinates $x^{a}$ are said to be distinguished by the foliation. Note that, according to this definition, $\mathrsfso{F}$ may be described locally in terms of pairs $(U_{i}, f_{i})$, where $U_{i}$ is an open subset in $M$, and $f_{i}:U_{i}\rightarrow \mathbb{R}^{m-k}$, is a submersion i.e., the differential $d_{p}f_{i}:T_{p}M\rightarrow T_{f_{i}(p)}\mathbb {R}^{m-k}$ is surjective in $U_{i}$. Therefore, the plaques of $\mathrsfso{F}$ in $U$ are the connected components of the sets $f^{-1}(c)$, $c\in\mathbb{R}^{m-k}$.

A j-form $\omega\in\Omega^{j}(M)$ is called \textit{simple} (or \textit{totally decomposable}\footnote{In the mathematical literature the second definition is standard; it goes back to
Cartan but is, sometimes, also called \textit{pure} or \textit{totally divisible}.}) if there exist 1-forms $\omega^{1},...,\omega^{j}\in\Omega^{1}(M)$ such that $\omega=\omega^{1}\wedge...\wedge\omega^{j}$. In what follows, we consider involutive distributions on $M$ induced by simple, closed, forms $\omega\in\Omega^{m-k}(M)$ and $\tau\in\Omega^{k}(M)$. The fact that simple closed $j$-forms characterize involutive distributions is obvious, since
\begin{equation}\label{property}
d\omega^{i}\wedge\omega=\omega^{i}\wedge d\omega=0,\quad\quad\quad d\tau^{j}\wedge\tau=\tau^{j}\wedge d\tau=0,
\end{equation}
for each $i=1,...,m-k$ and $j=1,...,k$ with $0\leq k\leq m$. We define these distributions by
\begin{equation}
D_{1}=\mbox{ker}\ \omega^{1}\cap ... \cap \mbox{ker}\ \omega^{m-k},\quad\quad\quad D_{2}=\mbox{ker}\ \tau^{1}\cap ... \cap \mbox{ker}\ \tau^{k},
\end{equation}
and denote the codimension-$(m-k)$ foliation by $\mathrsfso{F}_{1}$ and the codimension-$k$ foliation by $\mathrsfso{F}_{2}$, respectively.

An interesting geometrical situation emerges if $(M,\ g_{M})$ is a manifold with metric $g_{M}$ and we require that $\omega$ and $\tau$ are dual in the sense of the Hodge map
\begin{eqnarray*}
&&\star: \Omega^{m-k}(M)\rightarrow \Omega^{k}(M)\\
&&\omega\wedge\tau={s}<\omega,\star\tau>\varepsilon_{M}
\end{eqnarray*}
with $s$ the sign of the determinant of the matrix representation of $g_{M}$ on any basis, $<.\ ,\ .>$ the inner product induced by $g_{M}$ on the space of $(m-k)$-forms and $\varepsilon_{M}$ the volume-form. This leads us to the following definition.

\begin{defn}
Let $\omega\in\Omega^{m-k}(M)$ and $\tau\in\Omega^{k}(M)$ be simple and closed. If $\tau$ is the image of $\omega$ under the Hodge map i.e., $\tau=\star\omega$, we shall say that $\mathrsfso{F}_{1}$ and $\mathrsfso{F}_{2}$ define a $(m-k,k)$-\textit{dual foliation} of $(M,\ g_{M})$.
\end{defn}
\noindent The main property of such foliations is the fact that they are orthogonal i.e., at each point $p\in M$, if $u\in D_{1}(p)$ and $v\in D_{2}(p)$ then $g_{M}(u,v)=0$ (see, for instance, \cite{gravitation}). Also, the closedness of $\omega$ and $\tau$ implies in
\begin{equation}\label{equation}
d\omega =0\quad\quad\quad d(\star\omega)=0,
\end{equation}
which reduce to Maxwell equations in vacuum when $(M,\ g_{M})$ is ordinary spacetime and $k=2$. Indeed, we shall see that examples of such foliations appear naturally in the study of all simple solutions to Maxwell equations in vacuum. A similar result was first realized by Carter \cite{carter} in the restricted context of force-free plasmas and generalized by Uchida \cite{uchida} and Gralla \& Jacobson \cite{gralla}.

\section{Maxwell equations in vacuum and topological invariants}\label{Section III}
From now on we let $(M,\ g_{M})$ be a smooth 1+3 dimensional spacetime, with metric signature $(+,-,-,-)$, and work in natural units $c=1$. For the sake of concreteness, we assume $M$ to be globally hyperbolic and that its second cohomology group is trivial i.e., $H^{2}_{dR}(M)=0$ \cite{bott tu}. In a local chart $\{x^{a}\}$ $a=0,1,2,3$ the electromagnetic field and its dual $F, \star F\in\Omega^{2}(M)$ read as
\begin{equation}\nonumber
F=\frac{1}{2}F_{ab}\ dx^{a}\wedge dx^{b},\quad\quad \star F=\frac{1}{2}\star F_{ab}\ dx^{a}\wedge dx^{b},
\end{equation}
\noindent where $\star F_{ab}=\frac{1}{2}\varepsilon_{abcd}F^{cd}$ and
\begin{equation}\nonumber
\varepsilon_{abcd}=\sqrt{|g|}\tilde{\varepsilon}_{abcd},\quad\quad \varepsilon^{abcd}=-\frac{1}{\sqrt{|g|}}\tilde{\varepsilon}_{abcd},
\end{equation}
are the Levi-Civita tensors with $g\equiv \mbox{det}\ g_{ab}$ and $\tilde\varepsilon_{0123}=+1$. The two quadratic invariants are defined by
\begin{equation}\nonumber
\Psi=\frac{1}{2}F_{ab}F^{ab},\quad\quad\quad \Phi=\frac{1}{4}\star F_{ab}F^{ab},
\end{equation}
and, unless otherwise specified, we assume all fields to be of class $C^{\infty}$. With the above conventions, a direct computation reveals the algebraic relations
\begin{equation}\label{algeqs}
\star F^{ab}\star F_{bc}- F^{ab} F_{bc}=\Psi\delta^{a}_{\phantom a b},\quad\quad\quad F^{ab}\star F_{bc}=-\Phi\delta^{a}_{\phantom a b},
\end{equation}
and the top-degree identities
\begin{eqnarray}\label{1}
\star F\wedge \star F&=&-F\wedge F=2\Phi\ \varepsilon_{M}\\\label{2}
F\wedge \star F  &=&\star F\wedge  F=\Psi\ \varepsilon_{M}.
\end{eqnarray}
The latter quantities are nothing but the Chern-Pontrjagin densities, well known to gauge theorists \cite{jackiw}.

The source-free electromagnetic field is required to satisfy Maxwell equations in vacuum
\begin{equation}\label{Maxwell}
dF=0,\quad\quad\quad d(\star F)=0,
\end{equation}
which define a system of linear hyperbolic equations with constraints for the fields \cite{Christodoulou,geroch}. Now, as $H^{2}_{dR}(M)=0$, $F$ and $\star F$ will be globally exact, and we may write
\begin{equation}\label{Potential}
F=dA\quad\quad\quad \star F=dC
\end{equation}
with $A, C\in\Omega^{1}(M)$ the corresponding vector potentials. This leads us to the definition of the following 3-forms
\begin{equation}\nonumber
CS_{1}= F\wedge A,\quad\quad\quad CS_{2}= \star F\wedge C,\quad\quad\quad CS_{3}= F\wedge C,\quad\quad\quad CS_{4}= \star F\wedge A.
\end{equation} 
These are Chern-Simons-like terms which satisfy the relations
\begin{equation}
d(CS_{1})=-d(CS_{2})=-2\Phi\ \varepsilon_{M},\quad\quad\quad d(CS_{3})=d(CS_{4})=\Psi\ \varepsilon_{M},
\end{equation}
due to eq.(\ref{Maxwell}) and eq.(\ref{Potential}). Using Stokes theorem one concludes that
\begin{equation}
\int_{\partial\mathcal{V}}CS_{i}=\int_{\mathcal{V}}d(CS_{i})\quad\quad\quad i=1,2,3,4.
\end{equation}
with $\mathcal{V}$ a four-volume bounded by $\partial\mathcal{V}$. As a globally hyperbolic manifold $M$ is homeomorphic to $\mathbb{R}\times\Sigma$ where $\Sigma$ is a 3-manifold \cite{haw}, we may consider a volume of the form $\mathcal{V}=[t_{a},t_{b}]\times \Sigma\subset M$. Assuming that the fields decay sufficiently fast one obtains
\begin{eqnarray}
\mathcal{H}_{1}\equiv\int_{\Sigma}CS_{1}+CS_{2}=constant_{1},\\
\mathcal{H}_{2}\equiv\int_{\Sigma}CS_{3}-CS_{4}=constant_{2},
\end{eqnarray}
meaning that the integrals do not depend on the hypersurface as time evolves. It turns out that the quantities $\mathcal{H}_{1}$ and $\mathcal{H}_{2}$ are topological invariants intimately related to the \textit{electromagnetic helicities} introduced by Ra\~nada and Trueba in \cite{trueba}. This was first noticed by Elbistan \textit{et al} in \cite{elbi}, where the authors also show that the helicity is the moment map of duality acting as an $SO(2)$ group of canonical transformation on the symplectic space of all solutions of the vacuum Maxwell equations.
  
\section{Electromagnetic 2-forms and dual foliations}\label{Section IV}

\subsection{Simple solutions}

At a spacetime point $p$, any nonzero 2-form $\omega\in\Omega^{2}(M)$ is either simple $\omega=\omega^{1}\wedge\omega^{2}$, or can be written as $\omega=\omega^{1}\wedge \omega^{2}+\omega^{3}\wedge \omega^{4}$, with $\omega^{1}\wedge\omega^{2}\wedge\omega^{3}\wedge\omega^{4}\neq0$. In the former case $\mbox{dim}(\mbox{ker}\ \omega)=2$ while in the latter $\mbox{dim}(\mbox{ker}\ \omega)=0$. A necessary and sufficient condition for $\omega$ to be simple is that $\omega\wedge\omega=0$. One then concludes, using eq.(\ref{1}), that simple electromagnetic fields are uniquely defined by the algebraic condition $\Phi=0$. In what follows we deal only with simple solutions to eq.(\ref{Maxwell}) i.e., we assume that the electromagnetic 2-forms may be globally written as
\begin{equation}
F=\omega^{1}\wedge\omega^{2},\quad\quad\quad \star F=\tau^{1}\wedge\tau^{2}.
\end{equation}
The reason for this relies on the following useful proposition:

\begin{prop}
Every non-vanishing simple solution to Maxwell equations in vacuum induce a $(2,2)$-dual foliation of spacetime.
\end{prop}
\begin{proof}
The defining forms $\omega^{1}, \omega^{2}$ and $\tau^{1}, \tau^{2}$ characterize two rank-$2$ distributions
\begin{eqnarray*}
D_{1}&=&\mbox{ker}\ \omega^{1}\cap\mbox{ker}\ \omega^{2}\\
D_{2}&=&\mbox{ker}\ \tau^{1}\cap\mbox{ker}\ \tau^{2}
\end{eqnarray*}
These distributions are automatically involutive since
\begin{equation}\nonumber
d\omega^{1}\wedge F=0,\quad\quad d\omega^{2}\wedge F=0,\quad\quad d\theta^{1}\wedge \star F=0,\quad\quad d\theta^{2}\wedge \star F=0.
\end{equation}
They are also dual as $\star F$ is, by construction, the image of $F$ under the Hodge map, which completes the proof.
\end{proof}
 
Henceforth, we interpret simple solutions in terms of two codimension-2 foliations $\mathrsfso{F}_{1}$ and $\mathrsfso{F}_{2}$, whose leaves are the integral manifolds of $D_{1}$ and $D_{2}$. We adopt the following terminology: the leaves of $\mathrsfso{F}_{1}$ and $\mathrsfso{F}_{2}$ will be called \textit{magnetic leaves} and \textit{electric leaves}, respectively. We can think of these leaves as the worldsheets of the electromagnetic fields while they evolve in time. Here, the orthogonality of the leaves may be seen as a consequence of
\begin{equation}
\star F^{ac}F_{cb}=-\Phi\delta^{a}_{\phantom a b}=0.
\end{equation}
It is convenient to classify $(2,2)$-dual foliations in two classes, according to the signs of the invariants \cite{lichnero}:\\

\begin{enumerate}

\item{Non-null field ($\Psi\neq 0,\ \Phi=0$): $F$ has two distinct principal null directions and we may take $\omega^{1}\wedge\omega^{2}\wedge\tau^{1}\wedge\tau^{2}\neq0$, with\\

\begin{enumerate}	

\item{Magnetically dominated ($\Psi>0$):
\begin{eqnarray*}
\quad\quad\quad\quad\omega^{1}\cdot\omega^{1}<0,\quad\quad\omega^{2}\cdot\omega^{2}<0,\quad\quad \tau^{1}\cdot\tau^{1}<0,\quad\quad \tau^{2}\cdot\tau^{2}=1,
\end{eqnarray*}
with ``$\cdot$'' standing for the scalar product, for brevity. Here, the magnetic leaves are timelike whereas the electric leaves are spacelike.}\\ 

\item{Electrically dominated ($\Psi<0$):
\begin{eqnarray*}
\quad\quad\quad\quad\omega^{1}\cdot\omega^{1}<0,\quad\quad\omega^{2}\cdot\omega^{2}=1,\quad\quad \tau^{1}\cdot\tau^{1}<0,\quad\quad \tau^{2}\cdot\tau^{2}<0.
\end{eqnarray*}
In this case, the role of the leaves are interchanged. Magnetic leaves are spacelike whereas electric leaves are timelike.}\\ 

\end{enumerate}

We note that, for non-null solutions, $\mathrsfso{F}_{1}$ and $\mathrsfso{F}_{2}$ are not only orthogonal, but also transversal $\mathrsfso{F}_{1}\transv\mathrsfso{F}_{2}$ i.e., $T_{p}\mathrsfso{F}_{1}\oplus T_{p}\mathrsfso{F}_{2}=T_{p}M$.}\\

\item{Null field ($\Phi=0,\ \Psi=0$): the 2-form has one ``repeated'' principal null direction; We have:
\begin{eqnarray*}
&& \omega^{1}\cdot\omega^{1}<0,\quad\quad \omega^{2}\cdot\omega^{2}=0,\quad\quad  \tau^{1}\cdot\tau^{1}<0,\quad\quad  \tau^{2}\cdot\tau^{2}=0.
\end{eqnarray*}
with $\omega^{1}\wedge\omega^{2}\wedge\tau^{1}\neq 0$ but $\omega^{2}=\tau^{2}$. In this case, foliations are still orthogonal but $\mbox{dim}(T_{p}\mathrsfso{F}_{1} \oplus T_{p}\mathrsfso{F}_{2})=3$. The leaves have a null vector in common and a result originally due to Robinson \cite{robinson} implies that this vector is geodetic and shear-free.} 

\end{enumerate}

 
\subsection{Field lines}

A globally hyperbolic spacetime admits a congruence $\mathcal{C}_{v}$ of future-oriented timelike lines with unit tangent vector field $v$. Therefore, we can locally split the manifold into time plus space through the orthogonal decomposition of $TM$ on $\mathcal{C}_{v}$. The congruence $\mathcal{C}_{v}$ and the associated (local) 1+3 threading identifies a family of observers on $M$ whose world lines are the congruence curves and whose 4-velocity is $v$. Tensor fields which have no component along $v$ are called spatial with respect to $v$ and the subspace of $T_{p}M$ orthogonal to $v(p)$ is defined by $H_{p}=\mbox{ker}\ g(p)\circ v(p)$. The (1+3) splitting induces the following decompositions of the electromagnetic 2-forms
\begin{equation}
F_{ab}=E_{[a}v_{b]}+\varepsilon_{abcd}B^{c}v^{d}\quad\quad\quad \star F_{ab}=-B_{[a}v_{b]}+\varepsilon_{abcd}E^{c}v^{d}
\end{equation}
with $[ab]\equiv ab -ba$ denoting antisymmetrization and
\begin{equation}
E^{a}\equiv F^{a}_{\phantom a b}v^{b}\quad\quad\quad B^{a}\equiv-\star F^{a}_{\phantom a b}v^{b}
\end{equation}
Note that, according to these definitions, $E(p)$ and $B(p)$ are automatically elements of $H_{p}$ and, in a local Lorentz frame, we have\\
\begin{eqnarray*}
&&E^{1}=F^{10}=F_{01}=\star F^{32}=\star F_{32}\quad\quad\quad\quad B^{1}=F^{32}=F_{32}=\star F^{01}=\star F_{10}\\
&&E^{2}=F^{20}=F_{02}=\star F^{13}=\star F_{13}\quad\quad\quad\quad B^{2}=F^{13}=F_{13}=\star F^{02}=\star F_{20}\\
&&E^{3}=F^{30}=F_{03}=\star F^{21}=\star F_{21}\quad\quad\quad\quad B^{3}=F^{21}=F_{21}=\star F^{03}=\star F_{30}\\
\end{eqnarray*}
We interpret $E(p)$ and $B(p)$ as the electric and magnetic fields measured by the observer $v$ at the point $p$. A simple calculation reveals that the invariants read as
\begin{equation}
\Psi=B^{2}-E^{2}\quad\quad\quad \Phi=E\cdot B
\end{equation}
Notice also that $B\in \mbox{ker}\ F$ whereas $E\in \mbox{ker}\ \star F$ which justifies our terminology for the leaves in the case of simple fields.

The concept of space-filling lines of force appears in the following way. If we are given the observer $v^{a}$ and the vector fields $E^{a}$ and $B^{a}$, we may define their integral lines as
\begin{equation}\label{lines}
\frac{dx^{a}}{d\lambda}=E^{a}(x(\lambda))\quad\quad\quad \frac{dx^{a}}{d\sigma}=B^{a}(x(\sigma))
\end{equation}
with $\lambda,\sigma\in\mathbb{R}$. Although knowledge of these lines always into the qualitative behavior of fields, we need to be careful here.
\begin{enumerate}
\item{ Direction and intensity of field lines are highly dependent on the observer's choice.\\}

\item{Unless we impose additional constraints to $v$, field lines do not necessarily sit in a global hypersurface orthogonal to $v$.\\}

\item{Maxwell's equations written in terms of $E$ and $B$ become highly involved with the kinematical parameters of $v$ and therefore, the geometry/topology of the lines may look quite different for different observers.\\}
\end{enumerate}

\noindent A more tractable situation occurs if we impose the condition $v_{[a}\nabla_{b}v_{c]}=0$. In this case, the observer field is called \textit{Gaussian} and the subspaces $H_{p}$ mesh themselves to form a global space-like hypersurface $\Sigma$, orthogonal to $v^{a}$ everywhere. A natural question which arises here is the following: What is the relation between the dual foliations $\mathrsfso{F}_{1}$, $\mathrsfso{F}_{2}$, the hypersurface $\Sigma$ and the lines of force? The answer is given by the following theorem:

\begin{thm} \label{teo1}
Let $\mathrsfso{F}_{1}$ and $\mathrsfso{F}_{2}$ be the dual foliations induced by $F$ and $\star F$, respectively. If $\Sigma$ is a space-like hypersurface orthogonal to a Gaussian observer $v$, the lines of force of $B$ and $E$ as measured by $v$ coincide with the transversal intersections $\mathrsfso{F}_{1}\transv\Sigma$, $\mathrsfso{F}_{2}\transv\Sigma$, respectively.\\
\end{thm}

\begin{proof}
Write $F^{ab}=\omega_{1}^{[a}\omega_{2}^{b]}$ and $\star F^{ab}=\tau_{1}^{[a}\tau_{2}^{b]}$ for the contravariant components of the 2-forms $F=\omega^{1}\wedge\omega^{2}$ and $\star F=\tau^{1}\wedge\tau^{2}$, respectively. The magnetic and electric fields are given by
\begin{equation}\nonumber
B=(\tau_{1}\cdot v)\tau_{2}-(\tau_{2}\cdot v)\tau_{1},\quad\quad\quad E=(\omega_{2}\cdot v)\omega_{1}-(\omega_{1}\cdot v)\omega_{2}.
\end{equation}
\noindent At a point $p\in M$, the intersections $D_{1}\cap H_{p}$ and $D_{2}\cap H_{p}$ are determined by vectors $U=\lambda_{1}\tau_{1}+\lambda_{2}\tau_{2} \in D_{1}$ and $V=\chi_{1}\omega_{1}+\chi_{2}\omega_{2}\in D_{2}$, such that 
\begin{equation}
\lambda_{1}(\tau_{1}\cdot v)+\lambda_{2}(\tau_{2}\cdot v)=0,\quad\quad\quad \chi_{1}(\omega_{1}\cdot v)+\chi_{2}(\omega_{2}\cdot v)=0.
\end{equation}
\noindent Consider first the case $\tau_{2}\cdot v \neq 0$ and $\omega_{2}\cdot v \neq 0$. One obtains
\begin{equation}
U=-\frac{\lambda_{1}}{\tau_{2}\cdot v}B,\quad\quad\quad V=\frac{\chi_{1}}{\omega_{2}\cdot v}E.
\end{equation}
Therefore, the intersections $D_{1}\cap H_{p}$ and $D_{2}\cap H_{p}$ are one-dimensional subspaces of $H_{p}$ spanned by the vectors $B$ and $E$, respectively. We then have
\begin{equation}
D_{1}\oplus H_{p}=T_{p}M,\quad\quad\quad D_{2}\oplus H_{p}=T_{p}M.
\end{equation}
implying that the intersections are transversal. This is always the case for null fields and the general situation for non-null fields. For a magnetically dominated field, however, there exist an observer such that $\tau_{2}\cdot v \neq 0$ and $\omega_{2}\cdot v = 0$. For this observer one obtains
\begin{equation}
D_{1}\oplus H_{p}=T_{p}M,\quad\quad\quad D_{2}\oplus H_{p}=H_{p},
\end{equation}
which implies that $D_{2} \subseteq H_{p}$ and no electric field is present in the rest frame. Similarly, for the electric dominated case, there exist an observer such that
$\tau_{2}\cdot v = 0$ and $\omega_{2}\cdot v \neq 0$ and
\begin{equation}
D_{1}\oplus H_{p}=H_{p},\quad\quad\quad D_{2}\oplus H_{p}=T_{p}M,
\end{equation}
implying that $D_{1} \subseteq H_{p}$ and no electric field is present in the rest frame. If we extend these local results to the whole of $\Sigma$, we then conclude that the integral lines of $B$ ($E$) will exist if and only if $\mathrsfso{F}_{1}\transv\Sigma$ ($\mathrsfso{F}_{2}\transv\Sigma)$.
\end{proof}

\subsection{Fields defined by manifold mappings}

It is clear from definitions 1 and 2 that every $(2,2)$-dual foliation may be described locally in terms of a pair of spacetime submersions $\phi$, $\theta:M\rightarrow \mathbb{R}^{2}$. This means that there exist open neighborhoods $U\subseteq M$ such that magnetic leaves are described by the sets $\phi^{-1}(c)$, while the electric leaves are described by $\theta^{-1}(d)$ for $c,d\in\mathbb{R}^{2}$. 

We now turn ourselves to a different question: When does a pair of maps (preferably submersions) 
\begin{equation}\nonumber
\phi,\ \theta: (M,\ g_{M})\rightarrow (N,\ h_{N}),
\end{equation} for a bidimensional Riemannian manifold $(N,\ h_{N})$, define a (2,2)-dual foliation of $M$? The answer to this question is highly nontrivial since it depends firstly on the topologies of $M$ and $N$ and secondly on a solution to a system of fully nonlinear partial differential equations for the maps.

Note, however, that the pullbacks $\phi^{*}\varepsilon_{N}\in\Omega^{2}(M)$ and $\theta^{*}\varepsilon_{N}\in\Omega^{2}(M)$ of the volume form $\varepsilon_{N}$ are automatically simple and closed. Simplicity results from elementary algebraic reasons while closedness follows from the commutativity of $d$ with manifold mappings i.e., 
\begin{equation}\nonumber
d(\phi^{*}\varepsilon_{N})=\phi^{*}d\varepsilon_{N}=0,\quad\quad d(\theta^{*}\varepsilon_{N})=\theta^{*}d\varepsilon_{N}=0,
\end{equation}
since $\varepsilon_{N}$ is a top-degree form. Therefore, if we manage to find maps satisfying the \textit{dual condition}
\begin{equation}\label{dual condition}
\star(\phi^{*}\varepsilon_{N})=-\theta^{*}\varepsilon_{N}, \quad\quad\quad \star(\theta^{*}\varepsilon_{N})=\phi^{*}\varepsilon_{N},
\end{equation}
and define the electromagnetic 2-forms as
\begin{equation}
F\equiv -\phi^{*}\varepsilon_{N},\quad\quad\quad \star F\equiv \theta^{*}\varepsilon_{N},
\end{equation}
we do get a simple solution to Maxwell equations in vacuum, and consequently a (2,2)-dual foliation of spacetime.

In local charts $x^{a}:M\rightarrow\mathbb{R}^{1+3}$ and $y^{\alpha}:N\rightarrow\mathbb{R}^{2}$ the maps read as
\begin{equation}\nonumber
y^{\alpha}=\phi^{\alpha}(x)\quad\quad\quad y^{\alpha}=\theta^{\alpha}(x)
\end{equation}
with $\alpha,\beta=1,2$ while the electromagnetic fields take the form
\begin{eqnarray*}
F_{ab}&=&-\varepsilon_{\alpha\beta}(\phi(x))\partial_{a}\phi^{\alpha}\partial_{b}\phi^{\beta}\\
\star F_{ab}&=&+\varepsilon_{\alpha\beta}(\theta(x))\partial_{a}\theta^{\alpha}\partial_{b}\theta^{\beta}
\end{eqnarray*}
with $\varepsilon_{\alpha\beta}$ the Levi-Civita tensor compatible with $h_{N}$. The dual condition (\ref{dual condition}) then becomes
\begin{equation}
\varepsilon^{abcd}\varepsilon_{\alpha\beta}(\phi(x))\partial_{c}\phi^{\alpha}\partial_{d}\phi^{\beta}=-\varepsilon_{\alpha\beta}(\theta(x))\partial_{a}\theta^{\alpha}\partial_{b}\theta^{\beta}
\end{equation}
which are highly nonlinear equations. To the best of our knowledge no general method for solving equations of this type exists in literature (see \cite{Dacorogna}, however, for several results in the context of the so-called pullback equation for differential forms). Nevertheless, it is clear from this construction that, if a global solution exists, the corresponding magnetic/electric leaves will be described by the inverse images $\phi^{-1}(c)$ and $\theta^{-1}(d)$ for $c,d\in N$. In the case of submersive maps i.e., when 
\begin{eqnarray*}
&&d\phi_{p}:T_{p}M\rightarrow T_{\phi(p)}N\\
&&d\theta_{p}:T_{p}M\rightarrow T_{\theta(p)}N
\end{eqnarray*}
are surjective, the preimage theorem guarantees that the corresponding leaves (inverse images) are 2-dimensional submanifolds nicely embedded in $M$. In the case when $\phi,\  \theta$ have critical values, simple solutions to Maxwell equations may still exist. However, the inverse images may be rather complicated subsets of $M$.

As a final remark, we note that a volume form $\varepsilon_{N}$ has no local structure in the sense that it is not possible on small open sets to distinguish between itself and the volume form on Euclidean space \cite{kobayashi}, we can always choose local coordinates $(y^{1},y^{2})$ in $N$ such that $\varepsilon_{N}=dy^{1}\wedge dy^{2}$. If this choice is made, the above formalism reduces to the formalism of covariant Euler potentials described in \cite{uchida}.\\

Some important fields which can be obtained in this way are the following:\\
\begin{enumerate}
\item{\textbf{Spherically symmetric field}: Write $(t,r,\xi_{1},\xi_{2})$ for spherical coordinates in $\mathbb{R}^{1+3}$ and $(R,\Theta)$ for polar coordinates in $\mathbb{R}^{2}$. The line elements read as
\begin{equation}\nonumber
ds^{2}=dt^{2}-dr^{2}-r^{2}d\xi_{1}^{2}-r^{2}\mbox{sin}^{2}\xi_{1} d\xi_{2}^{2}
\end{equation}
\begin{equation}\nonumber
d\ell^{2}=dR^{2}+R^{2}d\Theta^{2}
\end{equation}
whereas the corresponding volume forms are given by
\begin{equation}\nonumber
\varepsilon_{M}=r^{2}\mbox{sin}\xi_{1}\ dt\wedge dr\wedge d\xi_{1}\wedge d\xi_{2} \quad\quad\quad \varepsilon_{N}=R\ dR\wedge d\Theta
\end{equation}
Consider the maps $\phi,\ \theta:\mathbb{R}^{1+3}\backslash\{r=0\}\rightarrow\mathbb{R}^{2}$, defined by:
\begin{equation}\nonumber
       \phi = 
        \begin{cases}
            R=r^{-1/2} & \\
            \Theta=-2et& 
        \end{cases}
    \quad\quad\quad 
    \theta = 
        \begin{cases}
            R=(\mbox{cos}\xi_{1})^{1/2} & \\
            \Theta=2e\xi_{2} & 
        \end{cases}
    \end{equation}
with $e\in\mathbb{R}$. The pullbacks read as
\begin{equation}\nonumber
\phi^{\ast}\varepsilon_{N}= -\frac{e}{r^{2}}\ dt\wedge dr,\quad \quad\quad\quad \theta^{\ast}\varepsilon_{N}=-e\ \mbox{sin}\xi_{1}\ d\xi_{1}\wedge d\xi_{2},
\end{equation}
which obviously satisfy the dual condition (\ref{dual condition}) and Maxwell's equations (\ref{Maxwell}). Here, magnetic leaves are spacelike 2-spheres centered at the origin (preimages of $\phi$) whereas the electric leaves are timelike 2 half-planes given by the equations $\xi_{1}=constant_{1},\ \xi_{2}=constant_{2}$ (preimages of $\theta$). Since all the magnetic leaves are submanifolds of $\mathbb{R}^{3}$ they can not intersect  the latter transversely. As a consequence of Theorem \ref{teo1} we conclude that no magnetic lines of force are present in this rest frame. In contrast, all electric leaves are transversal to $\mathbb{R}^{3}$. Theorem \ref{teo1} then implies that electric lines of force pointing in the radial direction are present, as expected. Observe that the situation drastically changes for a different observer, although the geometry/topology of the leaves remains intact.} \\  
\item{\textbf{Plane waves}: Write $(t,x,y,z)$ and $(X,Y)$ for cartesian coordinates in $\mathbb{R}^{1+3}$ and $\mathbb{R}^{2}$, respectively. The line elements read as
\begin{equation}\nonumber
ds^{2}=dt^{2}-dx^{2}-dy^{2}- dz^{2}
\end{equation}
\begin{equation}\nonumber
d\ell^{2}=dX^{2}+dY^{2}
\end{equation}
whereas the corresponding volume forms are given by
\begin{equation}\nonumber
\varepsilon_{M}= dt\wedge dx\wedge dy\wedge dz \quad\quad\quad \varepsilon_{N}=dX\wedge dY. 
\end{equation}
Consider the almost submersions $\phi,\ \theta:\mathbb{R}^{1+3}\rightarrow\mathbb{R}^{2}$:
\begin{equation}\nonumber
       \phi = 
        \begin{cases}
            X=E_{0}x & \\
            Y=\mbox{sin}[(kz-\omega t)]/k & 
        \end{cases}
    \quad\quad\quad 
    \theta = 
        \begin{cases}
            X=E_{0}y & \\
            Y=\mbox{sin}[(kz-\omega t)]/k& 
        \end{cases}
    \end{equation}
with $k,\omega, E_{0}>0$ and $\omega/k=1$. These maps are submersive, except at the zero measure set $\{ (t,x,y,z)\in\mathbb{R}^{1+3}|\ kz-\omega t = (n+1)\pi/2,\ n\in\mathbb{Z}\}$, and the pullbacks read as
\begin{eqnarray*}
\phi^{\ast}\varepsilon_{N}&=&E_{0}\mbox{cos}(kz-\omega t)dx\wedge d(z-t)\\
\theta^{\ast}\varepsilon_{N}&=&E_{0}\mbox{cos}(kz-\omega t)dy\wedge d(z-t)
\end{eqnarray*}
They also satisfy the dual condition (\ref{dual condition}) and Maxwell's equations (\ref{Maxwell}). \noindent Now, both magnetic and electric leaves are  2-dimensional planes containing a null-line. The magnetic leaves are given by the sets $x=constant_1 , z=t+constant_2$ and the electric leaves are given by $y=constant_3,z=t+constant_4$. The dual foliation associated to the plane wave, being a null solution, necessarily intersect the rest frame transversely. This means that both magnetic and electric lines always will be present.}
\end{enumerate}

\section{The Hopf-Ra\~nada solution: an example of a (2,2) linked dual foliation}\label{Section V}

In this section we apply the formalism of dual foliations to the Hopf-Ra\~nada solution. Although, strictly speaking, the solution is not knotted it is the simplest known solution endowed with a nontrivial topology: all field lines are linked once. However, most properties described here can be easily generalized to more complex knotted solutions, such as those described in \cite{dal, de klerk 1, dennis, kedia}.

\subsection{Orthogonal Hopf fiberings of $\mathbb{S}^{3}$}

To begin with, let $\mathbb{S}^{3}=\{(u_{1}+iu_{2},u_{3}+iu_{4})\in\mathbb{C}^{2}|u_{1}^{2}+u_{2}^{2}+u_{3}^{2}+u_{4}^{2}=1\}$ and interpret the 2-sphere as the compactified complex plane $\mathbb{S}^{2}\cong\mathbb{C}\cup\infty$ by introducing a complex coordinate $w$ via stereographic projection from the north pole onto the plane through the equator. Consider the set of surjective maps
\begin{eqnarray*}
&&\Pi_{k}:\mathbb{S}^{3}\rightarrow\mathbb{S}^{2}\quad\quad\quad k=1,2,3\\
&&(u_{1}+iu_{2},u_{3}+iu_{4})\mapsto w=\Pi_{k}(u_{1},u_{2},u_{3},u_{4})
\end{eqnarray*}
defined by:
\begin{equation} 
\Pi_{1}=\frac{u_{1}+iu_{2}}{u_{3}+iu_{4}},\quad\quad\quad \Pi_{2}=\frac{u_{2}+iu_{3}}{u_{1}+iu_{4}},\quad\quad\quad \Pi_{3}=\frac{u_{3}+iu_{1}}{u_{2}+iu_{4}}.
\end{equation}
Note that they are related by the cyclic permutations $u_{1}\rightarrow u_{2}\rightarrow u_{3}$, keeping the fourth coordinate $u_{4}$ fixed. $\Pi_{k}$ define three proper surjective submersions, since the differentials $d_{u}\Pi_{k}:T_{u}\mathbb{S}^{3}\rightarrow T_{\Pi_{k}(u)}\mathbb{S}^{2}$ have maximal rank for all $u\in\mathbb{S}^{3}$ and inverse images of compact subsets are compact. Therefore, due to Ehresmann fibration theorem, the maps $\Pi_{k}$ define three locally trivial fibrations $\mathbb{S}^{1}\hookrightarrow\mathbb{S}^{3}\xrightarrow{\Pi_{k}}\mathbb{S}^{2}$. We refer the reader to \cite{trueba1} for a nice derivation in terms of Lie groups.

The normalized area 2-form $\varepsilon_{\mathbb{S}^{2}}\in\Omega^{2}(\mathbb{S}^{2})$ reads as
\begin{equation}\label{area}
\varepsilon_{\mathbb{S}^{2}}=\frac{1}{2\pi i}\frac{dw\wedge d\overline{w}}{(1+w\overline{w})^{2}},
\end{equation} 
from which we construct the simple, closed, 2-forms $\mathfrak{f}^{k}\in \Omega(\mathbb{S}^{3})$, in terms of the following pullbacks:
\begin{equation}\nonumber
\mathfrak{f}^{k}\equiv\Pi_{k}^{*}\varepsilon_{\mathbb{S}^{2}}
\end{equation}
with $k=1,2,3$. A tedious, though straightforward, calculation gives\\
\begin{eqnarray*}
\mathfrak{f}^{1}&=&\left[-du_{1}\wedge du_{2}-(u_{2}/u_{4})du_{2}\wedge du_{3}+(u_{1}/u_{4})du_{3}\wedge du_{1}\right]/\pi,\\\\
\mathfrak{f}^{2}&=&\left[(u_{2}/u_{4})du_{1}\wedge du_{2}-du_{2}\wedge du_{3}-(u_{3}/u_{4})du_{3}\wedge du_{1}\right]/\pi,\\\\
\mathfrak{f}^{3}&=&\left[-(u_{1}/u_{4})du_{1}\wedge du_{2}+(u_{3}/u_{4})du_{2}\wedge du_{3}-du_{3}\wedge du_{1}\right]/\pi.\\
\end{eqnarray*}

From these pulled-back forms there follow the 1-forms $X^{k}\in\Omega^{1}(\mathbb{S}^{3})$
\begin{equation}\nonumber
X^{k}=\star \mathfrak{f}^{k},
\end{equation}
with $\star$ the Hodge dualization in $\mathbb{S}^{3}$. We leave for the reader to verify that they are automatically orthonormalized, i.e.
\begin{equation}
X^{i}\cdot X^{j}=\frac{1}{\pi^{2}}\delta^{ij}
\end{equation}
where $i,j=1,2,3$ "$\cdot$" is the scalar product in $\mathbb{S}^{3}$, and that
\begin{equation}\label{simp}
\mathfrak{f}^{1}=-\pi X^{2}\wedge X^{3}\quad\quad\quad \mathfrak{f}^{2}=\pi X^{1}\wedge X^{3}.
\end{equation}
Thus, the local defining forms for $\mathfrak{f}^{1}$ and $\mathfrak{f}^{2}$ are not linearly independent. This fact will play an important role in the next section. As a final remark, we note that composition of $X^{k}$ with the contravariant metric gives three orthogonal vector fields in $\mathbb{S}^{3}$, and it is easy to realize that the preimages $\Pi_{k}^{-1}(c)$, for $c\in\mathbb{S}^{2}$, coincide with the integral lines of these fields. As a consequence, the fibrations defined by $\Pi_{k}$ are orthogonal everywhere. Furthermore, as the maps $\Pi_{k}$ are nothing but Hopf maps, the fibres of each fibration are all linked once.

\subsection{From spacetime to the 3-sphere}

In order to produce electromagnetic fields from the fibrations described above we need to compose the submersions $\Pi_{k}$ with another map $\Omega:\mathbb{R}^{1+3}\rightarrow\mathbb{S}^{3}$. This map is engineered in such a way that:\\
\begin{enumerate}

\item{It effectively compactifies spacetime i.e., all points at spatial infinity are mapped to a single point in $\mathbb{S}^{3}$.}\\

\item{It is a submersion i.e., the differential $d_{p}\Omega: T_{p}\mathbb{R}^{1+3}\rightarrow T_{\Omega(p)}\mathbb{S}^{3}$ is surjective.}\\

\item{It satisfies the relation $\star(\Omega^{*}\mathfrak{f}^{1})=-\Omega^{*}\mathfrak{f}^{2}$\label{frak}.}\\
\end{enumerate}

\noindent It is clear that such a construction will characterize a null solution to Maxwell equations in vacuum if we make the identifications
\begin{equation}\label{EMF}
F=-\Omega^{*}\mathfrak{f}^{1},\quad\quad\quad \star{F}=\Omega^{*}\mathfrak{f}^{2},
\end{equation}
 as $\mathfrak{f}^{1}$ and $\mathfrak{f}^{2}$ are closed and the top degree forms $F\wedge F$, $F\wedge \star F$ vanish due to relations (\ref{simp}).

The map $\Omega$ has an additional remarkable property: \textit{its preimages define a congruence of null geodesics in spacetime}. Indeed, writing $u^{A}=\Omega^{A}(x)$ $A=1,2,3$ and defining the pullbacks
\begin{equation}\nonumber
\omega^{1}_{a}=\partial_{a}\Omega^{A}X_{A}^{2}\quad\quad \omega^{2}_{a}=\partial_{a}\Omega^{A}X_{A}^{3}\quad\quad\tau^{1}_{a}=\partial_{a}\Omega^{A}X_{A}^{1}
\end{equation}
we obtain, in components:
\begin{equation}\label{tauomega}
\varepsilon_{ab}^{\phantom a\phantom a cd}\omega^{1}_{c}\omega^{2}_{d}=-\tau^{1}_{[a}\omega^{2}_{b]}.
\end{equation}
A carefull analysis of (\ref{tauomega}) implies in the following algebraic relations
\begin{eqnarray*}\nonumber
&&\omega^{1}\cdot \omega^{2}=0\quad\quad \omega^{1}\cdot \tau^{1}=0\quad\quad \omega^{2}\cdot \tau^{1}=0\\
&&\omega^{1}\cdot\omega^{1}<0\quad\quad\tau^{1}\cdot\tau^{1}<0\quad\quad \omega^{2}\cdot \omega^{2}=0,
\end{eqnarray*}
as expected for a null solution. Contraction of (\ref{tauomega}) with $\tau_{1}^{b}$ also gives
\begin{equation}\label{oomega}
\varepsilon^{abcd}\tau^{1}_{b}\omega^{1}_{c}\omega^{2}_{d}=(\tau^{1}\cdot\tau^{1})\omega_{2}^{a}.
\end{equation}
Now, for a point $p\in\Omega^{-1}(u)$, $u\in\mathbb{S}^{3}$, a vector $X\in T_{p}\mathbb{R}^{1+3}$ will be tangent to the preimage if and only if it belongs to $\mbox{ker}\ d\Omega$, which is one-dimensional due to the preimage theorem. Contracting both sides of (\ref{oomega}) with the Jacobian matrix $\partial_{a}\Omega^{A}$, one obtains
\begin{equation}
(\tau^{1}\cdot\tau^{1})\partial_{a}\Omega^{A}\omega_{2}^{a}=(\varepsilon^{abcd}\partial_{a}\Omega^{A}\partial_{b}\Omega^{B}\partial_{c}\Omega^{C}\partial_{d}\Omega^{D})X^{1}_{B}X^{2}_{C}X^{3}_{D}=0,
\end{equation}
which vanishes, since it is proportional to the push-forward of a four-dimensional anti-symmetric object $\Omega_{*}\varepsilon$ and $\tau^{1}\cdot\tau^{1}\neq 0$. One then concludes
\begin{equation}
\mbox{ker}\ d\Omega =\mbox{span}(\omega_{2})
\end{equation}
implying that the preimages are null curves in spacetime. It turns out that such curves are also shear-free geodesics and, therefore, define the Robinson congruence of the solution. Although the rigorous proof of this statement is not difficult it requires the use of bicovariant equations, which are outside the scope of this article (see, however, \cite{goulart}).

The fact that the preimages are null straight lines has an in important consequence in our analysis: The preimage $\mathcal{S}=\Omega^{-1}(\gamma)$ of any curve $\gamma$ in $\mathbb{S}^{3}$ is a \textit{ruled surface} in spacetime. Recall that a surface $\mathcal{S}$ is ruled (also called a scroll) if through every point of $\mathcal{S}$ there is a straight line that lies on $\mathcal{S}$. In this case, $\gamma$ is called the \textit{ruled surface directrix} and the null geodesic a \textit{director line}. Consequently, $\mathcal{S}$ can be described as the set of points swept out by a moving null director line along the a base curve. 

\subsection{From spacetime to the 2-sphere} Since the electromagnetic fields defined in this section are simple, they characterize a (2,2) dual foliation of spacetime. The maps $\phi,\ \theta:\mathbb{R}^{1+3}\rightarrow\mathbb{S}^{2}$ are given as the map compositions
\begin{equation}
\phi=\Pi_{1}\circ\Omega\quad\quad\quad \theta=\Pi_{2}\circ\Omega
\end{equation}
which satisfy the dual condition (\ref{dual condition}) by construction, with $N=\mathbb{S}^{2}$. Let us investigate the properties of these foliations for the Hopf-Ra\~nada solution. In this case, the map $\Omega$ satisfying all the conditions above is chosen as:\\
 \[
    \Omega=\left\{
                \begin{array}{ll}
                  u_1=(Ax-tz)/(A^2+t^2)\\
                  u_2=(Ay+t(A-1))/(A^2+t^2)\\
                  u_3=(Az+tx)/(A^2+t^2)\\
                  u_4=(A(A-1)-ty)/(A^2+t^2)\\
                \end{array}
              \right.
  \]\\
with $A=\frac{1}{2}(r^2-t^2+1)$. A direct computation gives the well-known complex `scalar fields'
\begin{equation}
\phi=\frac{(Ax-tz)+i(Ay+t(A-1))}{(Az+tx)+i(A(A-1)-ty)}\quad\quad \theta=\frac{(Ay+t(A-1))+i(Az+tx)}{(Ax-tz)+i(A(A-1)-ty)},
\end{equation}
from which, using (\ref{area}) and (\ref{EMF}), we obtain the electromagnetic fields
\begin{equation}
F=\frac{1}{2\pi i}\frac{d\overline{\phi}\wedge d{\phi}}{(1+\phi\overline{\phi})^{2}},\quad\quad\quad\star F=\frac{1}{2\pi i}\frac{d\theta\wedge d\overline{\theta}}{(1+\theta\overline{\theta})^{2}}.
\end{equation}

Although the explicit expressions for the fields are messy, the foliations $\mathrsfso{F}_{1}$ and $\mathrsfso{F}_{2}$ are easily described by the the sets $\phi^{-1}(c)$ and $\theta^{-1}(d)$, with $c,d\in\mathbb{S}^{2}$, which are obviously submanifolds of spacetime. The magnetic leaves intersect the electric leaves orthogonally and, as the solution is null, we have $\mbox{dim}(T_{p}\mathrsfso{F}_{1} \oplus T_{p}\mathrsfso{F}_{2})=3$ i.e., the leaves share a null vector. Alternatively, we may think in $\mathrsfso{F}_{1}$ and $\mathrsfso{F}_{2}$ in terms of the composed preimages $\Omega^{-1}(\Pi^{-1}_{1}(c))$ and $\Omega^{-1}(\Pi^{-1}_{2}(d))$, respectively. Since these are preimages of linked circles in $\mathbb{S}^{3}$, we conclude that the magnetic/electric leaves are linked null surfaces in spacetime. It is clear that almost all magnetic leaves are homeomorphic to cylinders $\mathbb{S}^{1}\times\mathbb{R}$, except one which is homeomorphic to the plane $\mathbb{R}^{2}$ (same for electric leaves). The latter is precisely the leaf mapped to the south pole in $\mathbb{S}^{2}$ i.e., the leaf `through spatial infinity'. In a sense this the opposite situation of the Reeb foliation of $\mathbb{S}^{3}$, where almost all leaves are non-compact except one core leaf homeomorphic to a torus. 

It turns out that the nontrivial topologies of the foliations are directly related to the invariants defined in Section \ref{Section III}. In the case of a null solution $\mathcal{H}_{1}$ and $\mathcal{H}_{2}$ split each in two independently conserved quantities. Hence, the magnetic and electric helicities
\begin{eqnarray*}
h_{11}&=&\int_{\Sigma}CS_{1}=\int_{\Sigma}F\wedge A=c_{11}\\
h_{22}&=&\int_{\Sigma}CS_{2}=\int_{\Sigma}\star F\wedge C=c_{22}\\
\end{eqnarray*}
for constants $c_{11}$, $c_{22}$, and the cross helicities (see, for instance, \cite{hoyos})
\begin{eqnarray*}
h_{12}&=&\int_{\Sigma}CS_{3}=\int_{\Sigma}F\wedge C=c_{12}\\
h_{21}&=&\int_{\Sigma}CS_{4}=\int_{\Sigma}\star F\wedge A=c_{21}
\end{eqnarray*}
for constants $c_{12}$, $c_{21}$ and $\Sigma$ any spacelike hypersurface homeomorphic to $\mathbb{R}^{3}$. These integrals are obviously gauge invariant when the integrals are over all space, since all fields are compactly supported. Also, as all these fields are actually pullbacks of non-singular quantities defined in $\mathbb{S}^{3}$, the (2,2) dual foliations inherit properties of the fibrations $\Pi_{k}$. Indeed, we have, for example
\begin{equation}
h_{11}=\int_{\Sigma}\Omega^{*}(\mathfrak{f}^{1}\wedge a^{1})=\int_{\mathbb{S}^{3}}\mathfrak{f}^{1}\wedge a^{1}
\end{equation}
with $\mathfrak{f}^{1}=-da^{1}$, which reduces to the well-known Whitehead integral. Here, the last equality follows from the fact that $\Omega$ may be interpreted as a one-parameter family of diffeomorphisms of $\Sigma$ to $\mathbb{S}^{3}/\{south\ pole\}$ and the integrand is compactly supported, which permits an Alexandrov compactification of $\Sigma$ to obtain $\mathbb{S}^{3}$.

As a final remark we note that, as a by-product of the framework described above we also obtain a third foliation $\mathrsfso{F}_{3}$. Its leaves are defined as the preimages of the map $\psi:\mathbb{R}^{1+3}\rightarrow\mathbb{S}^{2}$
\begin{equation}
\psi=\Pi_{3}\circ\Omega=\frac{(Az+tx)+i(Ax-tz)}{(Ay+t(A-1))+i(A(A-1)-ty)}
\end{equation}
Interestingly, the distribution $D_{3}$ which generates $\mathrsfso{F}_{3}$ is defined by the null vector $\omega_{2}$ and a third spacelike vector $\zeta_{1}$, simultaneously orthogonal to $\omega_{1}$ and $\tau_{1}$. Therefore, the Poynting vector is necessarily contained in $D_{3}$. We shall analyze the properties of $\mathrsfso{F}_{3}$ in a future communication.

\section{Discussion}\label{Section VII}

In this paper we have investigated simple solutions to Maxwell equations in vacuum from the perspective of foliation theory. To do so, we have defined the concept of a $(k, m-k)$ dual foliation and studied several aspects of their geometric structures. The machinery used in the paper emphasizes intrinsic properties of the leaves instead of observer dependent quantities - such as a time function, a local rest frame or a Cauchy hypersurface. We think this approach is more akin to the perspective of relativists and may motivate mathematicians to deal with this sort of problems.

After proving how a field line emerge from a (2,2) dual foliation and showing how submersions engender electromagnetic fields, we move on to the analysis of electromagnetic knots. According to our analysis, null knotted solutions may be interpreted in terms of two codimension-2 foliations, whose knotted leaves intersect non-transversely, but orthogonally, in spacetime. This is in accordance with previous results due to \cite{kedia2}. We then use the Hopf-Ra\~nada solution to illustrate how our method works and clarify several aspects of its construction. In particular, we show that the magnetic/electric leaves are linked ruled surfaces and how the linking appears in an observer independent fashion.

It is quite remarkable that electromagnetic knots may be studied without making any reference to field lines. Hence, we hope that the point of view adopted here  may serve as a unifying framework for understanding the connections between the four known different constructions of knotted solutions: manifold mappings, twistors, Bateman pairs, conformal. Until now, only few advances have been made in this direction, as stressed by \cite{bouwmeester}. Another exciting prospects are the following: i) explore the conformal invariance of the theory to obtain knotted solutions in conformally flat spaces; ii) construct new knotted solutions seeking for different orthogonal fibrations of $\mathbb{S}^{3}$, such as the Seifert fibration; iii) investigate how the phenomena of helicity exchange and line reconnection fits in the scenario of dual foliations. We shall investigate these problems in a forthcoming paper.

\subsubsection*{Acknowledgments:} This work was developed at IMPA-RJ and CAP-UFSJ and was supported by Capes-Brasil.
The first author would like thanks the staff of CAP-UFSJ, for the hospitality and support.

\bibliographystyle{amsalpha}
{}

\end{document}